\documentclass[11 pt]{article}

\usepackage{amsmath}
\usepackage{amssymb}
\usepackage{amsthm}
\usepackage{hyperref}
\usepackage{mdwtab}
\usepackage{multirow}
\usepackage{mathabx}
\usepackage{algorithm2e}
\usepackage{bm}
\usepackage{tikz}
  \usetikzlibrary{arrows,automata}

\newtheorem{theorem}{Theorem}
\newtheorem{proposition}{Proposition}
\newtheorem{corollary}{Corollary}

\theoremstyle{definition}
\newtheorem{definition}{Definition}

\newcommand{\SOS}{\text{S1S}}
\newcommand{\SOSl}{\text{S1S}(\in,<)}
\newcommand{\SOSs}{\text{S1S}(\in,\bs{s})}
\newcommand{\ESOSl}{\exists\text{S1S}(\in,<)}
\newcommand{\ESOSs}{\exists\text{S1S}(\in,\bs{s})}

\newcommand{\ria}{\rightarrow}

\newcommand{\NN}{\mathbb{N}}

\newcommand{\sub}{\subseteq}
\newcommand{\bs}[1]{\boldsymbol{#1}}
\newcommand{\tvec}[2]{\begin{pmatrix} #1\\ #2 \end{pmatrix}}
\newcommand{\aut}{(Q,\Sigma,\Delta,q_{\mathit{initial}},F)}

\title{The existential fragment of $\SOS$ over $(\in,\bs{s})$ is the co-B\"uchi languages}
\author{Egor Ianovski}

\begin{document}
  \maketitle

  \begin{abstract}
   \noindent  B\"uchi's theorem, in establishing the equivalence between languages definable
    in $\SOS$ over $(\in,<)$ and the $\omega$-regular languages also demonstrated that
    $\SOS$ over $(\in,<)$ is no more expressive than its existential fragment. It is also easy 
    to see that $\SOS$ over $(\in,<)$ is equi-expressive with $\SOS$ over $(\in,\bs{s})$.
    However, it is not immediately obvious whether it is possible to adapt B\"uchi's argument
    to establish equivalence between expressivity in $\SOS$ over $(\in,\bs{s})$ and its
    existential fragment. In this paper we show that it is not: the existential fragment of
    $\SOS$ over $(\in,\bs{s})$ is strictly less expressive, and is in fact equivalent to the
    co-B\"uchi languages.
  \end{abstract}

  \section{Preliminaries}

  \subsection{Second order theory of one successor}

  \begin{definition}[S1S syntax]
    We introduce the following components:
    \begin{itemize}
      \item
        A set of first order variables, denoted by lower case letters, possibly with subscripts.
      \item
        A set of second order variables, denoted by upper case letters, possibly with subscripts.
    \end{itemize}
    $\SOSl$ has the following set of well formed formulae:
    $$\varphi::=t<t\ |\ t\in X\ |\ \varphi\wedge\varphi\ |\ \neg\varphi\ |\ \exists x_i\varphi\ |\ \exists X_i\varphi.$$
    $t$ is understood to range over terms, which in this case are just the first order variables, and $X$ to range over second order variables.

    $\ESOSl$ consists of those formulae with an initial block of second order quantifiers over a first order matrix:
    $$\varphi_1::=t<t\ |\ t\in X\ |\ \varphi\wedge\varphi\ |\ \neg\varphi\ |\ \exists x_i\varphi.$$
    $$\varphi::=\varphi_1\ |\ \exists X_i\varphi.$$

    In $\SOSs$ we have two cases for terms, letting $x$ range over first order variables:
    $$t::=x\ |\ \bs{s}t.$$
    $$\varphi::=t\in X\ |\ \varphi\wedge\varphi\ |\ \neg\varphi\ |\ \exists x_i\varphi\ |\ \exists X_i\varphi.$$

   $\ESOSs$ follows analogously:
    $$t::=x\ |\ \bs{s}t.$$
    $$\varphi_1::=t\in X\ |\ \varphi\wedge\varphi\ |\ \neg\varphi\ |\ \exists x_i\varphi.$$
    $$\varphi::=\varphi_1\ |\ \exists X_i\varphi.$$

  The abbreviations $\vee,\ria$ and $\forall$ are defined in the usual way.
  \end{definition}

  \begin{definition}[S1S semantics]
    A formula of S1S $\varphi(\overline{x_i},\overline{X_i})$ with $n$ free first order variables and $m$ free second order
    variables will be evaluated on $(\overline{a_i},\overline{A_i})$, with $a_i\in\NN$ and $A_i\in 2^\NN$. Satisfaction is
    defined inductively:
    \begin{align*}
      (\overline{a_i},\overline{A_i})\vDash\bs{s}^kx_i\in X_j&\iff a_i+k\in A_j.\\
      (\overline{a_i},\overline{A_i})\vDash x_i<x_j&\iff a_i<a_j.\\
      (\overline{a_i},\overline{A_i})\vDash\neg\varphi&\iff (\overline{a_i},\overline{A_i})\nvDash\varphi.\\
      (\overline{a_i},\overline{A_i})\vDash\varphi\wedge\psi&\iff (\overline{a_i},\overline{A_i})\vDash\varphi\text{ and }(\overline{a_i},\overline{A_i})\vDash\psi.\\
      (\overline{a_i},\overline{A_i})\vDash\exists x_i\varphi&\iff \exists b_i\in\NN:(\overline{a_i}[a_i\leftarrow b_i],\overline{A_i})\vDash\varphi.\\
      (\overline{a_i},\overline{A_i})\vDash\exists X_i\varphi&\iff \exists B_i\in2^\NN:(\overline{a_i},\overline{A_i}[A_i\leftarrow B_i])\vDash\varphi.\\
    \end{align*}
    A model will be represented by an infinite word over $\{0,1\}^{n+m}$, with the $i$th component being the characteristic word of the $i$th set (treating numbers as singleton sets). We will conflate a model with its representation in the sequel.
  \end{definition}

  Note that the set of models satisfying a formula $\varphi(\overline{x_i},\overline{X_i})$ thereby induce a language over $\{0,1\}^{n+m}$. The notion of
  an $L$ definable language, with $L\in\{\SOSl,\SOSs,\ESOSl,\ESOSs\}$, thus follows in the obvious way: a language is $L$ definable just if it describes the
  models of some formula $\varphi(\overline{x_i},\overline{X_i})$ of $L$.

  \begin{proposition}
    The $\SOSl$ definable languages are precisely the $\SOSs$ definable languages.
  \end{proposition}
  \begin{proof}
    Every instance of $x<y$ can be replaced with $\forall X(\forall z(\bs{s}z\in X\ria z\in X)\ria (y\in X\ria x\in X) )$. In every formula involving $\bs{s} t$, $\bs{s} t$ can be replaced by $x$ and $\forall y(t<y\ria(x<y\vee x=y))$.
  \end{proof}
  \subsection{Automata on infinite words}

    \begin{definition}[B\"uchi automata]
      A \emph{B\"uchi automaton} is a 5-tuple $A=\aut$, with $Q$ and $F$ finite sets, $\Delta\sub Q\times\Sigma\times Q$, $q_{\mathit{initial}}\in Q$ and $F\sub Q$.

      The language accepted by $A$ is the set of all $\alpha\in\Sigma^\omega$ such that there exists a $\rho\in Q^\omega$ such that $\rho_0=q_{initial}$, and
      $(\rho_i,\alpha_i,\rho{i+1})\in\Delta$ such that for infinitely many $i$, $\rho_i\in F$. A language accepted by some B\"uchi automaton is called \emph{B\"uchi recognisable}.
    \end{definition}

    We remind the reader of the following three standard results:
   
    \begin{theorem}
      The B\"uchi recognisable languages are precisely the $\omega$-regular languages.
    \end{theorem}

    \begin{theorem}[B\"uchi's theorem 1\footnote{Strictly speaking, we are not referring to the result in \cite{Buchi1990}, but rather to the way this result is presented in, for example,
    the University of Oxford course on Automata, Logic and Games.}]\label{thm:oneway}
      Every $\SOSl$ definable language is B\"uchi recognisable.
    \end{theorem}

    \begin{theorem}[B\"uchi's theorem 2]\label{thm:otherway}
      Every B\"uchi recognisable language is $\ESOSl$ definable.
    \end{theorem}

    The crux of the proof is that for every $A=\aut$ we can construct the following formula defining the same language:

    \begin{align*}
      \varphi_A(\overline{X_a})=\exists Y_1\dots\exists Y_n
      \begin{pmatrix}
      &\mathit{partition}(Y_1,\dots,Y_n)\\
      \wedge& \exists Z\exists x \big(x\in Z\wedge\forall y( sy\notin Z)\wedge x\in Y_1\big)\\
      \wedge& \forall x \bigvee_{(i,a,j)\in\Delta} (x\in Y_i\wedge x\in X_a\wedge x\in Y_j)\\
      \wedge& \bigvee_{q\notin F} \mathit{infinite}(Y_q)
      \end{pmatrix}.
    \end{align*}
    Intuitively, $i\in Y_j$ just if the automaton is in state $j$ at step $i$. The first line states that the automaotn is at precisely one state at every step,
    the second that the automaton starts in the initial state, the third that the transition relation is respected and the last that some accepting state is visited infinitely often.

    \begin{definition}[co-B\"uchi automata]
      A \emph{co-B\"uchi automaton} is a 5-tuple $A=\aut$, with $Q$ and $F$ finite sets, $\Delta\sub Q\times\Sigma\times Q$, $q_{\mathit{initial}}\in Q$ and $F\sub Q$.

      The language accepted by $A$ is the set of all $\alpha\in\Sigma^\omega$ such that there exists a $\rho\in Q^\omega$ such that $\rho_0=q_{initial}$, and
      $(\rho_i,\alpha_i,\rho{i+1})\in\Delta$ such that $\rho_i\notin F$ for finitely many $i$. A language accepted by some co-B\"uchi automaton is called \emph{co-B\"uchi recognisable}.
    \end{definition}

    \begin{proposition}
      Every co-B\"uchi recognisable language is $\ESOSs$ definable.
    \end{proposition}
    \begin{proof}
    Let $A=\aut$ be a co-B\"uchi automaton. Consider the following formula:
    \begin{align*}
      \varphi_A(\overline{X_a})=\exists Y_1\dots\exists Y_n
      \begin{pmatrix}
      &\mathit{partition}(Y_1,\dots,Y_n)\\
      \wedge& \exists Z\exists x \big(x\in Z\wedge\forall y( sy\notin Z)\wedge x\in Y_1\big)\\
      \wedge& \forall x \bigvee_{(i,a,j)\in\Delta} (x\in Y_i\wedge x\in X_a\wedge x\in Y_j)\\
      \wedge& \bigwedge_{q\notin F} \mathit{finite}(Y_q)
      \end{pmatrix}.
    \end{align*}
    The finite predicate is defined as follows:
    \begin{align*}
      \mathit{finite}(X)=\exists Y\big(\forall x (x\in Y\ria sx\in Y)\wedge\forall x(x\in Y\ria x\notin X)\big).
    \end{align*}
    Note that neither the finite nor the zero predicate in $\varphi_A(\overline{X_a})$ is
    bound by a universal quantifier, so they can be pulled out front.
    \end{proof}

    To conclude this section, in Figure~\ref{fig:relation} we illustrate both the known relations between the languages so far introduced and the
    result we will prove in the next section.

    \begin{figure}
    \begin{center}
    \begin{tikzpicture}[->,>=stealth',shorten >=1pt,auto,
                    semithick]
      \node (A) at (0,0) {$\SOSl$};
      \node (B) at (0,-4) {$\ESOSl$};
      \node (C) at (8,0) {$\SOSs$};
      \node (D) at (8,-4) {$\ESOSs$};
      \node (E) at (4,-2) {$\omega$-regular};
      \node (F) at (4,-5) {co-B\"uchi};

      \path (A) edge node [above right] {\emph{Theorem~\ref{thm:oneway}}} (E)
	     (E) edge node [below right] {\emph{Theorem~\ref{thm:otherway}}} (B)
	     (B) edge node [right] {\emph{Inclusion}} (A)
	     (A) edge node {\emph{Mutual definability}} (C)
	     (C) edge node {} (A)
	     (D) edge node [right] {\emph{Inclusion}} (C)
	     (F) edge [bend right] node [right] {\emph{Inclusion}} (E)
	     (D) edge node [below right] {\emph{Theorem~\ref{thm:Conrad}}} (F)
	     (F) edge node [right] {} (D);
    \end{tikzpicture}
    \end{center}
    \caption{Inclusion between languages studied.}\label{fig:relation}
    \end{figure}
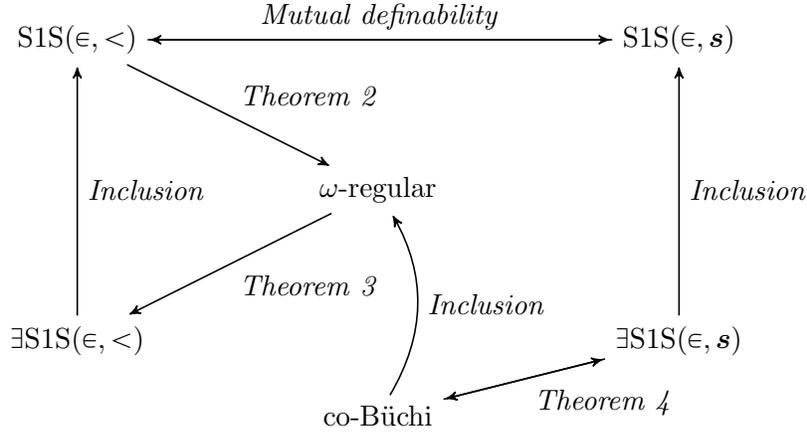

  \section{Proof}

    \begin{theorem}\label{thm:Conrad}
      The $\ESOSs$ definable languages are precisely the co-B\"uchi recognisable languages.
    \end{theorem}
    \begin{proof}
      We have already seen that every language accepted by a co-B\"uchi automaton can be
      defined with a $\ESOSs$ formula. It remains to show that for every $\ESOSs$ formula
      we can construct a co-B\"uchi automaton accepting the language it defines.

      Predictably, the construction will be given by induction. However, in contrast to the usual
      nature of such proofs the brunt of the work will fall on the base cases, while the inductive
      cases will follow trivially from the closure properties of co-B\"uchi automata.

      Let $\varphi=\exists X_1\dots\exists X_m\varphi'$ be a formula of $\ESOSs$, and $\varphi'$ 
      be its first order matrix. Observe that $\varphi'$ is a formula of monadic first order logic. As
      such, we can without loss of generality assume that:
      \begin{enumerate}
        \item
          $\varphi'$ is in negation normal form.
        \item
          Every occurrence of $\forall x_i$ has in its scope a quantifier-free disjunction of terms involving $x_i$
          and no other variables.\footnote{This follows from the fact that every formula of monadic first order logic is
          equivalent to a formula where every universal quantifier appears only in subformulae of the form:
          $$\forall x (F_1x\vee\dots\vee F_n x\vee \neg G_1 x\vee\dots\vee \neg G_m x)$$
          and every existential quantifier only in subformulae of the form:
          $$\exists x (F_1x\wedge\dots\wedge F_n x\wedge \neg G_1 x\wedge\dots\wedge \neg G_m x).$$}
      \end{enumerate}
      As a result of this assumption we will consider the following base cases:
      \begin{enumerate}
        \item
          $\bs{s}^kx_i\in X_j$, $x_i$ not in the scope of a universal quantifier.
        \item
          $\bs{s}^kx_i\notin X_j$, $x_i$ not in the scope of a universal quantifier.          
        \item
          $\forall x_i (\bs{s}^{k_1}x_i\in X_{j_1}\vee\dots\vee\bs{s}^{k_p}x_i\in X_{j_p}\vee\bs{s}^{k_1'}x_i\notin X_{j_1'}\vee\dots\vee\bs{s}^{k_q'}x_i\notin X_{j_q'}).$
      \end{enumerate}
      The inductive cases will be for $\vee$, $\wedge$, first order $\exists$ and second order $\exists$; the aim
      of the induction being to establish that for every formula $\varphi$, there exists a co-B\"uchi automaton
      $A_\varphi$ that accept $\alpha$ if and only if $\alpha$ is a model satisfying $\varphi$.
      To avoid a lengthy check at the end, we will instead verify that this property holds at the end of each
      case in the proof.

      \noindent{\bf  $\bs{s^kx_i\in X_j}$:}

      \noindent For a formula of the form $\bs{s}^kx_i\in X_j$, if $k\geq 1$ construct an automaton as in Figure~\ref{fig:base}
      \begin{figure}
      \begin{tikzpicture}[->,>=stealth',shorten >=1pt,auto,node distance=2.3cm,
                    semithick]
        \node[initial,state] (A)                    {$q_0$};
	\node[state]         (B) [right of=A] 	    {$q_1$};
	\node[state]         (D) [right of=B]       {$\dots$};
	\node[state]         (C) [right of=D]	    {$q_k$};
	\node[state,accepting](E) [below of=C]       {$q_{\mathit{accept}}$};
	\node[state]         (F) [right of=C]       {$q_{\mathit{dead}}$};

	\path (A) edge [loop below] node {$\tvec{0}{*}$} (A)
		  edge              node {$\tvec{1}{*}$} (B)
	      (B) edge  	    node {$\tvec{*}{*}$} (D)
	      (D) edge  	    node {$\tvec{*}{*}$} (C)
	      (C) edge              node {$\tvec{*}{0}$} (F)
		  edge              node {$\tvec{*}{1}$} (E)
	      (E) edge [loop left] node {$\tvec{*}{*}$} (E)
	      (F) edge [loop below]  node {$\tvec{*}{*}$} (F);
    \end{tikzpicture}
    \caption{Sketch of an automaton accepting $\bs{s}^kx_i\in X_j$.}\label{fig:base}
    \end{figure}
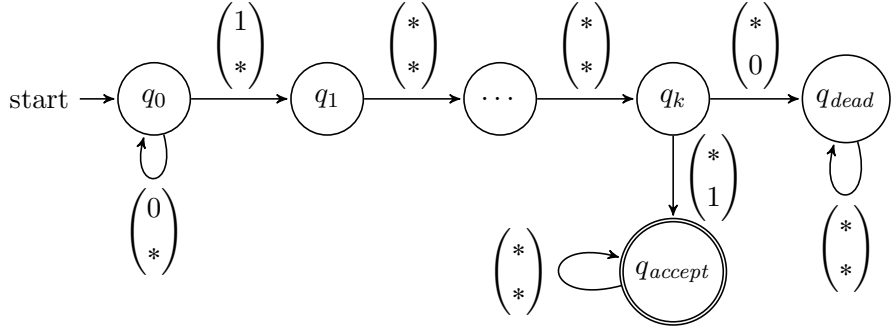

    That is to say, construct an automaton with $k+3$ states. Add a transition from $q_0$ to itself for every vector with a $0$ in the $i$th
    component, and a transition to $q_1$ for every vector with a $1$ in the $i$th component. Add a transition from $q_l$ to $q_{l+1}$
    for every vector. Add a transition from $q_k$ to $q_{final}$ for any vector with a 1 in the $n+j$th component, and a transition from
    $q_k$ to $q_{\mathit{dead}}$ for any vector with a 0 in the $n+j$th component. Add loops from $q_{\mathit{dead}}$ and $q_{\mathit{accept}}$ to themselves
    for any vector.

    In the case of $k=0$ construct an automaton with 2 states. $q_0$ has a transition to itself  for every vector with a $0$ in the $i$th
    component and a transition to $q_{\mathit{accept}}$ for every vector with a 1 in both the $i$th and $n+j$th component.

    We claim that this automaton has the following property: if $\alpha$ is a valid specification of a model\footnote{By which we mean the
    first order components contain one and only one 1 each.}, then the automaton accepts $\alpha$ if and only if it is a model satisfying
    $\bs{s}^kx_i\in X_j$.

    Suppose $\alpha$ is a word with well formed first order components. Upon reading the unique 1 in the $i$th component the
    automaton will transition to $q_1$, and $k$ steps later will enter the accepting state if it reads a 1 in the $n+j$th component, i.e.
    if $x_i+k$ is in $X_j$, as desired. In like manner, the only way the automaton can reach $q_{\mathit{accept}}$ is by reading a 1 in the $n+j$th
    component exactly $k$ steps after reading a 1 in the $i$th component, which would mean that if the input word is a model, it is a
    model satisfying $\bs{s}^kx_i\in X_j$.

    To transform this automaton into the $A_\varphi$ of the inductive hypothesis, simply intersect it with an automaton that accepts
    all valid model specifications (with $n$ first order and $m$ second order variables).

    \noindent{\bf  $\bs{s^kx_i\notin X_j}$:}

    \noindent For $k\geq 1$, construct the same automaton as before but make $q_{\mathit{dead}}$ the accepting state. For $k=0$ allow the transition
    to $q_{\mathit{accept}}$ only upon reading a 1 in the $i$th component and a 0 in the $n+j$th. Replicating the argument before, we obtain
    our result.

    \noindent{\bf  $\bs{\forall x_i}$:}

    \noindent This is the bulk of the proof, so we shall first offer some intuition.

    Suppose we have a formula $\varphi$ of the form:
    $$\forall x_i (\bs{s}^{k_1}x_i\in X_{j_1}\vee\dots\vee\bs{s}^{k_p}x_i\in X_{j_p}\vee\bs{s}^{k_1'}x_i\notin X_{j_1'}\vee\dots\vee
     \bs{s}^{k_q'}x_i\notin X_{j_q'}).$$
    Let $r$ be the number of distinct values of $k$ and $k'$, and label these values $\kappa_1,\dots,\kappa_r$ in increasing order.
    We wish to verify that every natural number satisfies $\varphi$, so let us take the perspective of the automaton and consider
    how we should verify whether 0 satisfies it. For this we must first wait for $\kappa_1$ steps, at which point there will be a number of
    terms of the form $\bs{s}^{\kappa_1}x_i\in X_{j}$ or $\bs{s}^{\kappa_1}x_i\notin X_{j}$. If the input word satisfies any one of these
    requirements, then 0 satisfies $\varphi$. If not, we cannot safely reject yet and have to wait for $\kappa_2-\kappa_1$ to check the
    next set of terms. Only if we reach $\kappa_r$ without finding any term the input word satisfies can we conclude that 0 does not
    satisfy $\varphi$ and hence reject the word.

    There is no problem, then, in designing an automaton that verifies whether any given integer satisfies $\varphi$. 
    The trouble is that we cannot pause or backtrack. While we are waiting on the outcome of 0, the entries corresponding to 1, 2, 3, will 
    have been read and we need to verify them concurrently,  and na\"ively constructing an automaton for each natural number will mean
    an infinite number of automata. However, we do not need an automaton for every single natural number: observe that if a number does
    satisfy $\varphi$, we shall know it within at most $\kappa_r$ states, as such we only need to concern ourselves with finitely many (precisely, $\kappa_r$)
    numbers at any given step.

    This leads us to the proof idea. We shall construct $\kappa_r$ automata, such that the $s$th automaton verifies that $s\mod\kappa_r$ satisfies $\varphi$. 
     By taking the intersection of all these automata the resulting machine will verify that all natural numbers satisfy $\varphi$.

    We will thus construct an automaton verifying $s\mod\kappa_r$. This is illustrated in Figure~\ref{fig:mod}.

    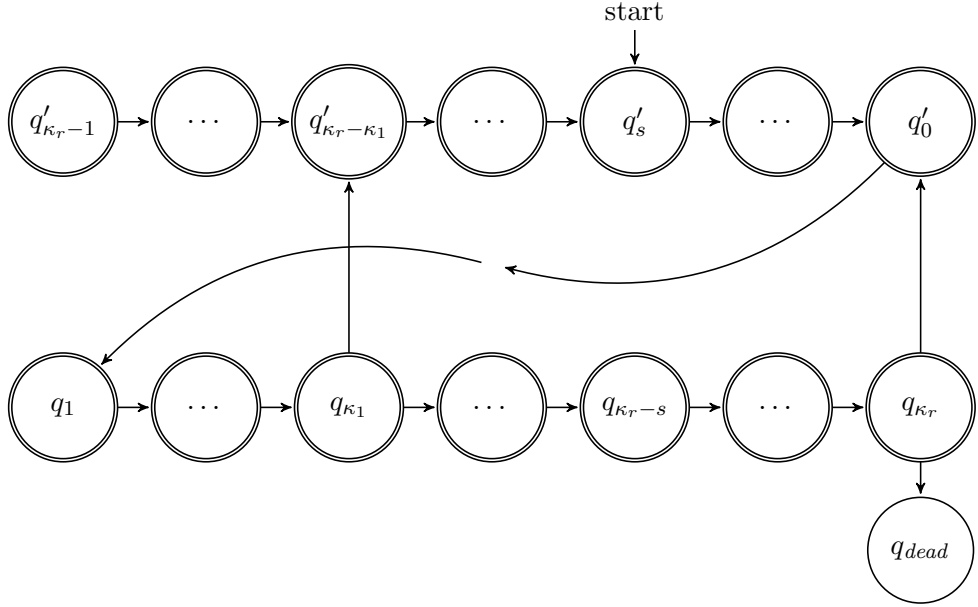
\begin{figure}
      \begin{tikzpicture}[->,>=stealth',shorten >=1pt,auto,node distance=1.9cm,
                    semithick]
        \tikzset{every state/.style={minimum size=40pt}}
        \node[state, accepting] (A)                    {$q_{\kappa_r-1}'$};
	\node[state, accepting]         (B) [right of=A] 	    {$\dots$};
	\node[state, accepting]         (C) [right of=B]       {$q_{\kappa_r-\kappa_1}'$};
	\node[state, accepting]         (D) [right of=C] 	    {$\dots$};
	\node[initial above,state, accepting]         (E) [right of=D] 	    {$q_s'$};
	\node[state, accepting]         (F) [right of=E] 	    {$\dots$};
	\node[state, accepting]         (G) [right of=F] 	    {$q_0'$};
	\node        (X) [below of=A] 	    {};
	\node        (Y) [below of=D] 	    {};
	\node[state, accepting]         (H) [below of=X] 	    {$q_1$};
	\node[state, accepting]         (I) [right of=H] 	    {$\dots$};
	\node[state, accepting]         (J) [right of=I]       {$q_{\kappa_1}$};
	\node[state, accepting]         (K) [right of=J] 	    {$\dots$};
	\node[state, accepting]         (L) [right of=K] 	    {$q_{\kappa_r-s}$};
	\node[state, accepting]         (M) [right of=L] 	    {$\dots$};
	\node[state, accepting]         (N) [right of=M] 	    {$q_{\kappa_r}$};
	\node[state]         (O) [below of=N] 	    {$q_{\mathit{dead}}$};

	\path (A) edge node {}(B)
	      (B) edge  	    node {} (C)
	      (C) edge  	    node {} (D)
	      (D) edge              node {} (E)
	      (E) edge              node {} (F)
	      (F) edge              node {} (G)
	      (G) edge[bend left]              node {} (Y)
	      (Y) edge[bend right]              node {} (H)
	      (H) edge              node {} (I)
	      (I) edge              node {} (J)
	      (J) edge              node {} (K)
	      (K) edge              node {} (L)
	      (L) edge              node {} (M)
	      (M) edge              node {} (N)
	      (J) edge              node {} (C)
	      (N) edge              node {} (G)
	      (N) edge              node {} (O);
    \end{tikzpicture}
    \caption{Sketch of an automaton accepting a word if all $s\mod\kappa_r$ satisfy the universal case.}\label{fig:mod}
    \end{figure}

    The automaton has states $\{q_0',\dots,q_{\kappa_r-1}'\}$, $\{q_1,\dots,q_{\kappa_r}\}$ and a dead state. All states
    sans the dead state are accepting. The initial state is $q_s'$. There is a transition from $q_{i+1}'$ to $q_i'$ for any vector,
    a transition from $q_0'$ to $q_1$ for any vector and a transition from $q_{u\neq\kappa_v}$ to $q_{u+1}$ for any vector.
    
    At $q_{\kappa_v}$ add a transition to $q_{\kappa_r-\kappa_v}'$ for any vector that satisfies one of the terms with $\kappa_v$
    successor operations. For example, if the terms with $\kappa_v$ successor terms in $\varphi$ are $\bs{s}^{\kappa_v}x_i\in X_3$ and
    $\bs{s}^{\kappa_v}x_i\notin X_4$ then add a transition to $q_{\kappa_r-\kappa_v}'$ for any vector with a 1 in the $n+3$th component or
    a 0 in the $n+4$th. For any other vector, add a transition from $q_{\kappa_v}$ to $q_{\kappa_v+1}$. If $\kappa_v=\kappa_r$, then
    add a transition to $q_{\mathit{dead}}$ instead.

    It is easy to see that if $\alpha$ satisfies $\varphi$, then since each $s\mod\kappa_r$ satisfies $\varphi$ the associated automaton
    will never leave the accepting zone, hence the intersection automaton will accept $\alpha$. If however $\alpha$ does not satisfy
    $\varphi$ then there must exist some number which will force an automaton to enter the dead state, causing the intersection automaton
    to reject.

    \noindent{\bf  $\bs{\vee}$:}

    \noindent Given a formula of the form $\varphi=\psi_1\vee\psi_2$, obtain $A_{\psi_1}$ and $A_{\psi_2}$ via the inductive hypothesis
    and let $A_\varphi=A_{\psi_1}\cup A_{\psi_2}.$

    Now, should $\alpha$ be a model satisfying $\varphi$ we can without loss of generality suppose it satisfies $\psi_1$. By inductive hypothesis
    $A_{\psi_1}$ accepts $\alpha$, hence $A_\varphi$ accepts $\alpha$. Similarly, should $\alpha$ be
    accepted by $A_\varphi$, it must be accepted by either $A_{\psi_1}$ or $A_{\psi_2}$ as required.

    \noindent{\bf  $\bs{\wedge}$:}

    \noindent Given a formula of the form $\varphi=\psi_1\wedge\psi_2$, obtain $A_{\psi_1}$ and $A_{\psi_2}$ via the inductive hypothesis
    and let $A_\varphi=A_{\psi_1}\cap A_{\psi_2}$. The desired property follows, mutatis mutandis, in the same manner as above.

    \noindent{\bf  $\bs{\exists x_i}$:}

    \noindent Given a formula of the form $\varphi=\exists x_i\psi$, obtain $A_\psi$ via the inductive hypothesis and project away the $i$th component
    to obtain $A_\varphi$. That is, $A_\varphi$ accepts $\alpha$ if and only if there is a way to reinsert the $i$th component into $\alpha$ such that
    the resulting $\alpha'$ is accepted by $A_\psi$.

    Suppose $\alpha$ is a model satisfying $\varphi$. It follows there must be a way to instantiate $x_i$ to satisfy $\psi$, meaning $\alpha$ can be expanded
    into $\alpha'$, which is accepted by $A_\psi$ by the induction hypothesis; in the case, $A_\varphi$ accepts $\alpha$ as desired. If  $A_\varphi$ accepts $\alpha$,
    then by definition of projection it must be the case that $\alpha$ can be expanded into an $\alpha'$ accepted by $A_\psi$, meaning $\alpha'$ is a model
    satisfying $\psi$ and hence $\exists\psi$ is satisfied by $\alpha'$ with the $x_i$ component ignored, namely $\alpha$.

    \noindent{\bf  $\bs{\exists X_i}$:}

    \noindent As before, but this time projecting away a second order variable.
    \end{proof}
    \begin{corollary}
      The languages definable by monadic first order logic over $(\in, \bs{s})$ on the natural numbers are included in the intersection of co-B\"uchi and deterministic B\"uchi recognisable languages.
    \end{corollary}
    \begin{proof}
      The reader will first observe that non-determinism is only introduced by the existential quantifiers. In the case of first order $\exists$, this could
      be avoided by giving a similar construction as in the $\forall$ case.  Next, note that in the universal construction the rejecting state is a sink,
      whereas in the existential construction the accepting state would be a sink. Because of this, there is no difference between the B\"uchi and co-B\"uchi
      acceptance conditions.
    \end{proof}
\bibliographystyle{plain}
\bibliography{references}
\end{document}